\documentclass[11pt]{article}

\usepackage{amssymb} 
\usepackage{amsmath}     
\usepackage{amsthm}

\usepackage{hyperref}

\usepackage{a4wide}

\usepackage{graphicx}

\newtheorem{thm}{Theorem}

\newenvironment{keyword}{\par{\noindent\bf Keywords:}}

\begin{document}

\title{Complexity of the robust weighted independent set problems on interval graphs}

\author{Adam Kasperski\footnote{Corresponding author}\\
   {\small \textit{Institute of Industrial}}\\
  {\small \textit{Engineering and Management,}}\\
  {\small \textit{Wroc{\l}aw University of Technology,}}\\
  {\small \textit{Wybrze{\.z}e Wyspia{\'n}skiego 27,}}\\
  {\small \textit{50-370 Wroc{\l}aw, Poland,}}\\
  {\small \textit{adam.kasperski@pwr.edu.pl}}
  \and  Pawe{\l} Zieli{\'n}ski \\
    {\small \textit{Institute of Mathematics}}\\
  {\small \textit{and Computer Science}}\\
  {\small \textit{Wroc{\l}aw University of Technology,}}\\
  {\small \textit{Wybrze{\.z}e Wyspia{\'n}skiego 27,}}\\
  {\small \textit{50-370 Wroc{\l}aw, Poland}}\\
{\small \textit{pawel.zielinski@pwr.edu.pl}}
}

\maketitle

\begin{abstract}
This paper deals with
 the max-min and min-max regret versions of the maximum weighted independent set problem on interval
 graphs
 with uncertain vertex weights.
 Both problems have been recently investigated by Nobibon and Leus (2014), who showed that they are NP-hard for two scenarios and strongly NP-hard if the number of scenarios
is a part of the input.
In this paper, new complexity and approximation results on the problems under consideration
are provided,
 which
extend the ones previously obtained. Namely,
 for the discrete scenario uncertainty representation
it is proven that if the number of scenarios $K$ is a part of the input, 
then the max-min version of the problem is not at all approximable. On the other hand,  its min-max regret version is approximable within $K$ and not approximable within $O(\log^{1-\epsilon}K)$ for any $\epsilon>0$ unless  the problems in NP have quasi polynomial algorithms. Furthermore, for the interval uncertainty representation it is shown that the min-max regret version is NP-hard and approximable within 2.
\end{abstract}


\section{Introduction}

We are given  a family $\mathcal{I}=\{I_1,I_2,\dots,I_n\}$ of closed intervals of real line, where $I_i=[a_i, b_i]$, $i\in [n]$
(we use  $[n]$ to denote the set $\{1,\ldots n\}$). An undirected graph $G=(V,E)$ with $|V|=n$ vertices and $|E|=m$ edges is called an \emph{interval graph} for $\mathcal{I}$ if $v_i\in V$ corresponds to $I_i$ and there is an edge $(v_i, v_j)\in E$ if and only if the intervals $I_i$ and $I_j$ have nonempty intersection. An \emph{independent set} $X$ in $G$ is a subset of the vertices of $G$ such for any $v_i,v_j\in X$ it holds $(v_i,v_j)\notin E$. We will use $\Phi$ to denote the set of all independent sets in $G$. For each vertex $v_i\in V$ a nonnegative weight $w_i$ is specified. In 
the \emph{maximum weighted independent set} problem
($\textsc{IS}$ for short), we seek an independent set~$X$ in $G$ of the maximum total weight $F(X)=\sum_{v_i \in X} w_i$. Contrary to the problem in general graphs, \textsc{IS} for interval graphs is polynomially solvable~\cite{PB96}. It has some important practical applications and  we refer the reader to~\cite{NL14, SPP07} for a description of them.

In~\cite{NL14} the following robust versions of the $\textsc{IS}$ problem have been recently investigated. Suppose 
that the vertex weights are uncertain and they are specified as a  scenario set $\Gamma$. Namely, each scenario 
$S\in \Gamma$ is a vector $(w_1^S,\dots,w_n^S)$ of nonnegative integral vertex weights which may occur. Now the weight of a solution~$X$ depends on a scenario and we will denote it by $F(X,S)=\sum_{v_i\in X} w_i^S$. Let $F^*(S)=\max_{X\in \Phi} F(X,S)$ be the weight of a maximum weighted independent set in~$G$ under scenario~$S$. In this paper, we wish to investigate the following two robust problems:
$$\textsc{Max-Min IS}:\;\; opt_1=\max_{X\in \Phi} \min_{S\in \Gamma} F(X,S),$$
$$\textsc{Min-Max Regret IS}:\;\; opt_2=\min_{X\in \Phi} \max_{S\in \Gamma}(F^*(S)- F(X,S)).$$
There are two popular methods of defining scenario set $\Gamma$ (see, e.g.,~\cite{KY97, K08}). For 
the \emph{discrete uncertainty representation} set $\Gamma=\{S_1,\dots,S_K\}$ contains $K$ explicitly given scenarios. For the  \emph{interval uncertainty representation}, for each vertex $v_i$ an interval $[\underline{w}_i, \overline{w}_i]$ of its possible weights is specified and $\Gamma$ is the Cartesian product of all these intervals. 

Both uncertainty representations have been studied  in~\cite{NL14}, where it has been shown that for the discrete uncertainty representation the max-min and min-max regret versions of the  $\textsc{IS}$ problem are NP-hard when $K=2$ and strongly NP-hard when the number of scenarios~$K$ is a part of input. 
Furthermore, some pseudopolynomial algorithms for both problems, when $K$ is constant, have been provided. 
For the interval uncertainty representation, the \textsc{Max-Min IS} problem  can be trivially reduced to a deterministic polynomially solvable counterpart, 
but the complexity of   \textsc{Min-Max Regret IS}  remained open. 
\paragraph{Our results}
We extend the complexity results obtained in~\cite{NL14} 
and provide new approximation ones
for both discrete and interval uncertainty representations. 
 Namely,
 for the discrete scenario uncertainty representation, we establish
that when the number of scenarios~$K$ is a part of the input, the \textsc{Max-Min IS} problem
is not at all approximable, the  \textsc{Min-Max Regret IS}  problem is approximable within~$K$ and not approximable within $O(\log^{1-\epsilon}K)$ for any $\epsilon>0$ unless  NP$\in$DTIME$(n^{\mathrm{ polylog}\, n})$.
We also show that both  \textsc{Max-Min IS} and \textsc{Min-Max Regret IS}
have \emph{fully polynomial-time approximation schemes} (FPTAS's), when $K$ is constant, i.e. for every
constant $\epsilon>0$
they admit $(1+\epsilon)$-approximation algorithms that run 
in
polynomial time both in $1/\epsilon$ and the size of their inputs.
  Furthermore, for the interval uncertainty representation, we prove that \textsc{Min-Max Regret IS} 
   is NP-hard and approximable within~2.

\section{Complexity and approximation results}

In this section, we extend the complexity results 
for the  \textsc{Max-Min~IS} and \textsc{Min-Max Regret~IS} problems
provided in the recent paper~\cite{NL14} 
and give new approximation ones
for both discrete and interval uncertainty representations. 
We start by considering the discrete scenario  uncertainty representation.
\begin{thm}
If $K$ is a part of the input, then \textsc{Max-Min IS} is strongly NP-hard and not at all approximable unless $P=NP$.
\label{tmmisap1}
\end{thm}
\begin{proof}
We provide a polynomial time reduction from the following 
the \textsc{Vertex Cover} problem, which is known to be strongly NP-complete~\cite{GJ79}. We are given 
an undirected graph~$G=(V,E)$, $|V|=n$,  and an integer~$L$.
 A subset of the vertices $W\subseteq V$ is a \emph{vertex cover} of~$G$ if for each $(v,w)\in E$ either $v\in W$ or $w\in W$ (or both). We ask if there is a vertex cover $W$ of $G$ such that $|W|\leq L$.
  We now construct an instance of \textsc{Max-Min IS} as follows. We first create 
  a family of intervals $\mathcal{I}=\{I_{ij}\}$, $i\in [n]$, $j\in [L]$, where $I_{ij}=[2j,2j+1]$ for each $i\in [n]$ and $j\in [L]$.
   It is easy to check that the resulting  interval graph $G'$ corresponding to~$\mathcal{I}$ is composed of~$L$ separate cliques of size~$n$ and each maximal independent set in~$G'$ contains exactly~$L$ vertices, one from every clique 
   (see Figure~\ref{fig1} - note that the  intervals from~$\mathcal{I}$ refer to the vertices in~$G'$ ). 
   \begin{figure}
\begin{center}
\includegraphics{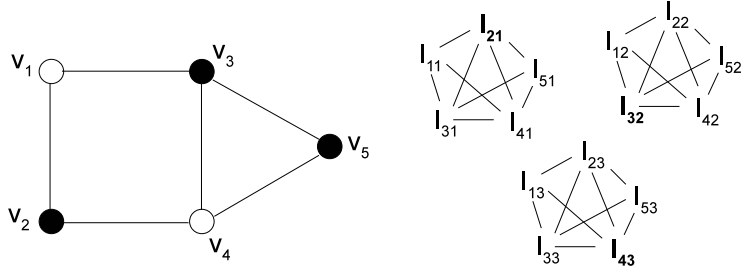}
\caption{An instance of \textsc{Vertex Cover} for $L=3$ and the corresponding interval graph~$G'$.}
\label{fig1}
\end{center}
\end{figure}
 We now form scenario set~$\Gamma$ as follows. For each edge $(v_k, v_l) \in E$,
    we create scenario under which the weights of 
    intervals (resp. vertices)
    $I_{kj}$ and $I_{lj}$ are equal to~1 for each $j\in [L]$ and the weights of the remaining intervals 
    (resp.  vertices) equal~0 (see Table~\ref{tab1}).
\begin{table}[ht]
\caption{Scenario set~$\Gamma$ for the instance from Figure~\ref{fig1}. Independent set $X=\{I_{21}, I_{32}, I_{53}\}$ corresponds to the vertex cover $W=\{v_2, v_3, v_5\}$ of size 3.} \label{tab1}
\begin{scriptsize}
\begin{center}
\begin{tabular}{l|cccccc}
& (1,2) & (1,3) & (2,4) & (3,4) & (3,5) & (4,5) \\ \hline \hline
$I_{11}$ & 1& 1& 0 & 0 & 0& 0\\
$\pmb{I_{21}}$ & \textbf{1}& \textbf{0}& \textbf{1} & \textbf{0}& \textbf{0}& \textbf{0}\\
$I_{31}$ & 0& 1& 0 & 1& 1& 0\\ 
$I_{41}$ & 0& 0& 1& 1& 0& 1\\
$I_{51}$ & 0& 0 & 0& 0& 1& 1\\ \hline
$I_{12}$ & 1& 1& 0 & 0& 0& 0\\ 
$I_{22}$ & 1& 0& 1 & 0& 0& 0\\
$\pmb{I_{32}}$ & \textbf{0}& \textbf{1}& \textbf{0} & \textbf{1}& \textbf{1}& \textbf{0}\\
$I_{42}$ & 0& 0& 1 & 1& 0& 1\\ 
$I_{52}$ & 0 & 0& 0 & 0& 1& 1\\ \hline
$I_{13}$ & 1 & 1& 0 & 0& 0& 0\\
$I_{23}$ & 1 & 0 & 1& 0& 0& 0\\ 
$I_{33}$ & 0& 1 & 0 & 1& 1& 0 \\
$I_{43}$ & 0& 0 & 1  & 1& 0& 1 \\
$\pmb{I_{53}}$ & \textbf{0}& \textbf{0}& \textbf{0}   & \textbf{0}& \textbf{1} & \textbf{1}\\
\end{tabular}
\end{center}
\end{scriptsize}
\end{table}

Suppose that there is a vertex cover~$W$ of~$G$ such that $|W|\leq L$.
 We lose nothing by assuming that $|W|=L$.
 One can easily meet this
assumption by adding arbitrary additional vertices to~$W$ - if necessary. 
 Hence $W=\{v_{i_1},\dots,v_{i_L}\}$. Let us choose an independent set $X$
 consists of the vertices in~$G'$ that correspond to intervals:
 $I_{i_11},\dots,I_{i_LL}$. 
 From the construction of the scenario set, it follows that $F(X,S)\geq 1$ for all $S\in \Gamma$ and, consequently, $opt_1\geq 1$. Assume now that $opt_1\geq 1$. So, there is an independent set~$X$ in~$G'$ such that $F(X,S)\geq 1$ under each scenario $S\in \Gamma$. 
 The independent set $X$ consists of the vertices corresponding to the intervals:
 $I_{i_11},\dots,I_{i_LL}$. 
 Consider the set of vertices $W=\{v_{i_1},\dots,v_{i_L}\}$. From the construction of the scenario set $\Gamma$ 
 it may be concluded that each edge of~$G$ is covered by~$W$. Therefore, $W$ is a vertex cover of size at most~$L$.
 Hence, the answer to \textsc{Vertex Cover} is yes if and only if $opt_1\geq 1$, which together with the fact that $opt_1\geq 0$,  imply 
  the non-approximability of \textsc{Max-Min IS}.
\end{proof}

In order to establish the hardness result for the \textsc{Min-Max Regret IS}, we will use the following variant of the \textsc{Label Cover} problem
(see e.g., \cite{AC95, MNO13}):
\begin{description}
\item[\mdseries \scshape Label Cover:]

We are given a regular  bipartite graph $G=(V\cup W,E)$, $E\subseteq V\times W$;
a set of labels $[N]$ and
for each edge $(v,w)\in E$ a  map (partial) $\sigma_{v,w}:[N]\rightarrow [N]$.
A \emph{labeling} of~$G$ is an assignment of a subset of labels to each of the vertices of $G$, i.e. a function  $l: V\cup W \rightarrow 2^{[N]}$. We say that a labeling \emph{satisfies}
an edge $(v,w)\in E$ if there exist $a\in l(v)$ and $b\in l(w)$ such that $\sigma_{v,w}(a)=b.$
A \emph{total labeling} is a labeling that satisfies all edges.
We seek a labeling whose  value defined by  $\max_{x\in V\cup W}|l(x)|$ is minimal.
This minimal value is denoted by 
$val(\mathcal{L})$, where $\mathcal{L}$ is the input instance.
\end{description}
\begin{thm}[\cite{MNO13}]
There exists a constant~$\gamma>0$ such that
for any language $L\in NP$, any input $\mathbf{w}$ and any $N>0$,
one can construct a \textsc{Label Cover} instance~$\mathcal{L}$  with the following properties in time polynomial in the instance's size:
\begin{itemize}
\item the number of vertices in $\mathcal{L}$ is $|\mathbf{w}|^{O(\log N)}$,
\item if $\mathbf{w}\in L$, then $val(\mathcal{L})=1$,
\item if $\mathbf{w}\not\in L$, then $val(\mathcal{L})> N^{\gamma}$.
\end{itemize}
\label{tlancover}
\end{thm}
The following theorem  will be needed in proving a lower bound on the approximation of \textsc{Min-Max Regret IS}.
\begin{thm}
There exists a constant~$\gamma>0$ such that
for any language $L\in NP$, any input $\mathbf{w}$, and any $N>0$,
one can construct an instance of
\textsc{Min-Max Regret IS} with the following properties:
\begin{itemize}
\item if $\mathbf{w}\in L$, then $opt_2\leq 1$,
\item if $\mathbf{w}\not\in L$, then $opt_2\geq \lfloor N^\gamma \rfloor:=g$,
\item the number of intervals is at most $|\mathbf{w}|^{O(\log N)}N$ and the number of scenarios is at most $|\mathbf{w}|^{O(g\log N)}N^{g}$.
\end{itemize}
\label{taprminmax}
\end{thm}
\begin{proof}
Suppose $L\in NP$
and let  $\mathcal{L}=(G=(V\cup W,E), N,\sigma)$ be the constructed instance of \textsc{Label Cover} for $L$
 (see Theorem~\ref{tlancover}). 
 We now build  a corresponding instance of \textsc{Min-Max Regret IS} in the following way. 
 We first  number the  edges of $G$ from~1 to $|E|$ 
 in arbitrary way. Then, for each edge $(v,w)\in E$,  we create a family of at most $N$ intervals 
$\mathcal{I}_{v,w}=\{I_{v,w}^{i,j}\,:\,\sigma_{v,w}(i)=j, i\in[N]\}$.
 If $(v,w)$ has a number $r\in \{1,\ldots, |E|\}$, then all the intervals in $\mathcal{I}_{v,w}$ are equal to $[2r, 2r+1]$. 
 We set $\mathcal{I}=\cup_{(v,w)\in E} \mathcal{I}_{v,w}$. It is easily seen that the corresponding
  interval graph~$G'$ for $\mathcal{I}$ is composed of exactly~$|E|$ separate cliques and each maximal independent set in this graph contains exactly $|E|$ intervals, one from each clique. 
  Note that the intervals from  $\mathcal{I}$ refer to the vertices in~$G'$.
Fix vertex $v\in V$. For each $g$-tuple of pairwise distinct edges $(v,w_1),\dots,(v,w_g)$ incident to $v$ and for each $g$-tuple of intervals $(I_{v,w_1}^{i_1,j_1},\dots,I_{v,w_g}^{i_g,j_g})\in \mathcal{I}_{v,w_1}\times\dots\times \mathcal{I}_{v,w_g}$,  where the labels $i_1,\dots,i_g$ are pairwise distinct,
 we form scenario under which all these intervals 
 (resp. the vertices in~$G'$)
 have the weight equal to~0 and all the remaining intervals 
  (resp. the vertices in~$G'$)
 have the weight equal to~1. 
We  proceed in this way  for each vertex $v\in V$.
 Choose vertex $w\in W$. For each $g$-tuple of pairwise distinct edges $(v_1,w),\dots,(v_g,w)$ incident to~$w$
  and for each $g$-tuple of intervals $(I_{v_1,w}^{i_1,j_1},\dots,I_{v_g,w}^{i_g,j_g})\in \mathcal{I}_{v_1,w}\times\dots\times \mathcal{I}_{v_g,w}$,  where the labels $j_1,\dots,j_g$ are pairwise distinct, we form scenario under which all these intervals
  (resp. the vertices)
   have the weight equal to~0 and all the remaining intervals 
   (resp. the vertices)
   have the weight equal to~1. We repeat this construction for each vertex $w\in W$.
  Finally, we add one scenario under which each vertex in~$G'$ has the weight equal to~1.
  We ensure in this way that the scenario set formed is not empty.
An easy computation shows that  in the above instance of \textsc{Min-Max Regret IS}
the cardinality of set~$\mathcal{I}$  is at most  $|E| N$ and the cardinality of the scenario set~$\Gamma$
 is at most  $|V||W|^gN^g+|W||V|^gN^g+1$.
 Hence and from the fact that
  the number of vertices (and also edges) in $G$ is $|\mathbf{w}|^{O(\log N)}$ (see Theorem~\ref{tlancover}),
  we have that $|\mathcal{I}|$ is at most $|\mathbf{w}|^{O(\log N)}N$ and 
  $|\Gamma|$ is at most $|\mathbf{w}|^{O(g\log N)}N^{g}$.

Assume now that $\mathbf{w}\in L$. 
Hence,
 there exists a total labeling~$l$, which assigns exactly one label~$l(v)$ to each $v\in V$ and exactly one label~$l(w)$
  to each $w\in W$. 
  Let us choose the interval~$I_{v,w}^{l(v),l(w)}\in \mathcal{I}_{v,w}$
   for each  $(v,w)\in E$. The vertices that refer to these intervals form 
  an independent set~$X$ in $G'$.
    There is at most one interval (vertex) with~0 weight 
    under each scenario,  and so
    $F(X,S)\geq |E|-1$ under each~$S \in \Gamma$.  
    Since $F^*(S)=|E|$ for each $S\in \Gamma$,  $opt_2\leq 1$.
     Suppose that $\mathbf{w}\notin L$, which gives $val(\mathcal{L})>N^{\gamma}$ and, in consequence,  $val(\mathcal{L})> \lfloor N^{\gamma}\rfloor=g$. 
     Assume, on the contrary, that $opt_2<g$.
     Thus,  there is an independent set~$X$ in $G'$ such that $F(X,S)>|E|-g$ under each scenario $S\in \Gamma$.   
     Note that~$X$ corresponds to a total labeling~$l$ which assigns labels~$i$ to~$v$ and $j$ to~$w$ 
     when the interval $I_{u,v}^{ij}$ is selected from~$\mathcal{I}_{v,w}^{i,j}$. 
     From the construction of~$\Gamma$,
     we conclude that $l$ assigns less than~$g$ distinct labels to each vertex~$x\in V\cup W$, 
     since otherwise $F(X,S)=|E|-g$ for some scenario $S\in\Gamma$. 
     Hence, we get $val(\mathcal{L})< g$, a contradiction.
\end{proof}
\begin{thm} 
If $K$ is a part of the input, then 
\textsc{Min-Max Regret IS} is not approximable within $O(\log^{1-\epsilon} K)$, for any $\epsilon>0$, unless NP$\in$DTIME$(n^{{\rm polylog}\, n})$
\end{thm}
\begin{proof}
Let $\gamma$ be the constant from Theorem~\ref{taprminmax}. Consider a language $L\in$~NP and an
 input~$\mathbf{w}$. Fix any constant $\beta>0$ and set  $N=\lceil \log^{\beta/\gamma} |\mathbf{w}|\rceil$.
 Theorem~\ref{taprminmax} allows us to  construct an instance of \textsc{Min-Max Regret IS} with the number of scenarios  $K$  asymptotically bounded by $|\mathbf{w}|^{\alpha N^\gamma \log N}N^{N^\gamma}$ for some constant $\alpha>0$, $opt_2\leq 1$ if $\mathbf{w}\in L$ and $opt_2\geq\lfloor \log^\beta |\mathbf{w}| \rfloor$ if $\mathbf{w}\notin L$.
 We get $\log K \leq \alpha N^\gamma \log N \log|\mathbf{w}|+N^\gamma \log N\leq \alpha' \log^{\beta+2}|\mathbf{w}|$ for some constant $\alpha'>0$ and sufficiently large $|\mathbf{w}|$. Therefore, $\log|\mathbf{w}|\geq (1/\alpha') \log^{1/(\beta+2)}K$ and the gap is at least $\lfloor \log^\beta |\mathbf{w}| \rfloor\geq \lfloor 1/\alpha' \log^{\beta/(\beta+2)} K \rfloor$.
The constant~$\beta>0$ can be arbitrarily large, and so the gap is $O(\log^{1-\epsilon} K)$ for any $\epsilon=2/(\beta+2)>0$.
Note that, the instance of \textsc{Min-Max Regret~IS} can be built in $O(|\mathbf{w}|^{{\rm poly log }|\mathbf{w}|})$ time, which completes the proof.
\end{proof}

We now show  that \textsc{Min-Max Regret~IS}  admits an approximation algorithm
with some guaranteed worst case ratio - contrary to \textsc{Max-Min~IS}, which is not at all approximable,
when~$K$ is a part of the input  (see Theorem~\ref{tmmisap1}).
Namely, there exists a simple $K$-approximation algorithm, which outputs an optimal solution
to the deterministic~\textsc{IS} problem with  
 the vertex weights computed as follows: $\hat{w}_i:=\frac{1}{K}\sum_{k\in [K]} w_i^{S_k}$, $i\in [n]$.
This can be done in $O(Kn+T(n))$ time, where $T(n)$ is the time for solving the deterministic~\textsc{IS} 
problem (e.g., $T(n)=O(n\log n)$, see~\cite{SPP07}).
\begin{thm}
 \textsc{Min-Max Regret IS} is approximable within~$K$. 
 \label{tkapp}
\end{thm}
\begin{proof}
  The proof is adapted from \cite[the proof of Proposition~4]{ABV10} to \textsc{Min-Max Regret IS}.
	Let $\hat{w}_i=\frac{1}{K}\sum_{k\in [K]} w_i^{S_k}$ be the average weight of vertex~$v_i\in V$
	 over all scenarios. Let $X^*$ be an optimal solution to \textsc{Min-Max Regret IS} and let $\hat{X}$ be an optimal solution for the deterministic weights~$\hat{w}_i$, $i\in [n]$. Clearly, $\hat{X}$ can be computed in 
	 polynomial time. 
	 The following inequalities hold:
	 $opt_2=\max_{k\in [K]} (F^*(S_k)-F(X^*,S_k))\geq \frac{1}{K}\sum_{k\in [K]} (F^*(S_k)-F(X^*,S_k))
	 \geq \frac{1}{K}\sum_{k\in [K]} (F^*(S_k)-F(\hat{X},S_k))
	 \geq\frac{1}{K}\max_{k\in [K]} (F^*(S_k)-F(\hat{X},S_k)).$
	 Hence the maximum regret of~$\hat{X}$ is at most $K\cdot opt_2$.

\begin{figure}[htbp]
\begin{center}
\includegraphics[height=6cm]{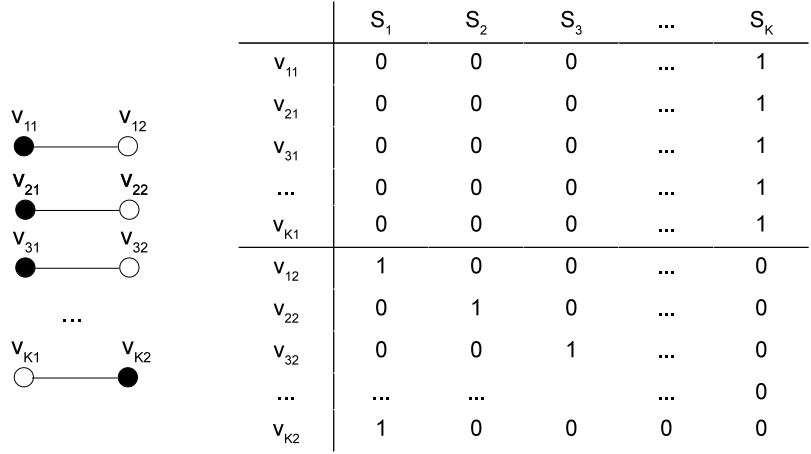}
\caption{A hard example for the $K$-approximation algorithm. The optimal independent set is marked in black.}
\label{fig4}
\end{center}
\end{figure}

To see that the bound is tight consider a sample problem shown in Figure~\ref{fig4}, where an interval graph composed of $2K$ vertices and the corresponding scenario set with $K$ scenarios are shown. The average weight of each vertex equals $1/K$. Hence the algorithm may return the independent set $\hat{X}=\{v_{12},v_{22},v_{32},\dots,v_{K2}\}$ whose maximal regret is equal to $K$. But the maximal regret of the independent set $X^*=\{v_{11},v_{21},v_{31},\dots,v_{K-1,1},v_{K2}\}$ is equal to $1$. 
\end{proof}

It turns out that  \textsc{Max-Min~IS} and \textsc{Min-Max Regret IS}  have  FPTAS's,
when $K$ is constant.
\begin{thm}
        If $K$ is constant, then both 
        \textsc{Max-Min~IS} and \textsc{Min-Max Regret IS}        
        admit  FPTAS's.
        \label{tfptas}
\end{thm}
\begin{proof}
The fact that \textsc{Max-Min~IS} admits an FPTAS
follows from \cite[Theorem~1]{ABV10} and the existence of the pseudopolynomial 
algorithm for this problem, provided in~\cite{NL14}, whose running time can be expressed by a  polynomial in
$w_{\max}$ and~$n$, where 
$w_{\max}=\max_{i\in[n], k\in[K]}w^{S_k}_i$.
An FPTAS for \textsc{Min-Max Regret IS}  is a consequence of \cite[Theorem~2]{ABV10}
and 
the existence of the pseudopolynomial 
algorithm for  \textsc{Min-Max Regret IS}, built in~\cite{NL14},
whose running time can be expressed by a  polynomial in~$U$ and~$n$,
where~$U$ is an upper bound on~$ opt_2$ such that $U\leq K\cdot L$ and $L$ is 
 a lower bound on~$ opt_2$.
 Of course,
 such lower and upper bounds can be provided by executing the $K$-approximation algorithm (see Theorem~\ref{tkapp}).
\end{proof}

We now turn to the interval uncertainty representation. 
For a given solution $X\in \Phi$, let~$S_X$ be the scenario under which the weights of $v_i\in X$ are $\underline{w}_i$ 
and the weights of $v_i \notin X$ are $\overline{w}_i$ for $i\in [n]$. 
It has been shown in~\cite{NL14} that the maximal regret of~$X$ is  $Z(X)=F^*(S_X)-F(X,S_X)$. This property
will be useful in proving 
the next two results. The first theorem gives  an answer 
to a question about the complexity of 
 \textsc{Min-Max Regret~IS} under  the interval uncertainty representation (only \textsc{Max-Min~IS} has been known to be polynomially solvable~\cite{NL14}, so far).
\begin{thm}
\textsc{Min-Max Regret IS} under interval uncertainty representation is NP-hard.
\end{thm}
\begin{proof}
We show a polynomial time reduction from the following \textsc{Partition} problem which is known to be NP-complete~\cite{GJ79}. We are given a collection $\mathcal{C}=(a_1,\dots,a_n)$ of positive integers. We ask if there is a subset $I\subseteq [n]$ such that $\sum_{i\in I} a_i=\sum_{[n]\setminus I} a_i$. Let us define $b=\frac{1}{2}\sum_{i\in [n]} a_i$. We now build the corresponding instance of \textsc{Min-Max Regret IS} as follows. The family of intervals $\mathcal{I}$ contains two intervals $I_{i1}=I_{i2}=[2i,2i+1]$ for each $i\in [n]$ and one interval $J=[1,2n+1]$. The corresponding interval graph for $\mathcal{I}$ is shown in Figure~\ref{fig2}.  The intervals from~$\mathcal{I}$ and  the interval~$J$
refer to the vertices in~$G$.

\begin{figure}[htbp]
\begin{center}
\includegraphics{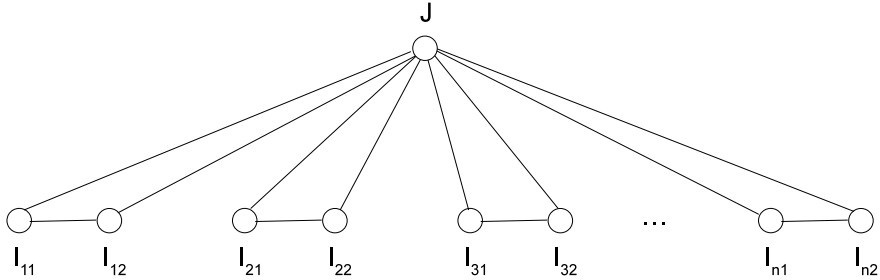}
\caption{The interval graph for the reduction.}
\label{fig2}
\end{center}
\end{figure}

Observe that each maximal independent set in~$G$ contains either one vertex~$J$ or exactly~$n$ vertices, one from each $I_{i1}, I_{i2}$, $i\in [n]$.
The interval  weight of~$I_{i1}$ is equal to $[3b-\frac{3}{2}a_i,3b]$, the interval  weight of $I_{i2}$ is equal to $[3b-a_i,3b-a_i]$, and the interval  weight of vertex $J$ is $[0,3nb-b]$. We now show that the answer to \textsc{Partition} is `yes' if and only if there is an independent set $X$ is $G$ such that $opt_2\leq \frac{3}{2}b$. 

Suppose that the answer to \textsc{Partition} is `yes'. Let $I\subseteq [n]$ be such that $\sum_{i\in I} a_i=\sum_{i\notin I} a_i=b$. Let us form an independent set $X$ in $G$ by choosing the vertices $I_{i1}$ for $i\in I$ and $I_{i2}$ for $i\notin I$.  It holds $F(X,S_X)=\sum_{i\in I} (3b-\frac{3}{2}a_i)+\sum_{i\notin I} (3b-a_i)=3nb-\frac{5}{2}b$
and $F^*(S_X)=\max\{3nb-b, \sum_{i\in I} 3b+\sum_{i\notin I} (3b-a_i)\}=3nb-b$. Hence $Z(X)= 3nb-b-3nb+\frac{5}{2}b=\frac{3}{2}b$.

Assume now that  $opt_2\leq \frac{3}{2}b$, so there is an independent set $X$ in $G$ such that $Z(X)\leq \frac{3}{2}b$. It must be $X\neq \{J\}$ since $Z(\{J\})\geq 3nb$. Hence $X$ is formed by the vertices $I_{i1}$ and $I_{i2}$ for $i \in [n]$. From the construction of graph $G$ it follows that $X$ contains either $I_{i1}$ or $I_{i2}$ for each $i\in [n]$ (but not both). Let $I$ be the subset of $[n]$ such that $I_{i1}\in X$ for each $i\in I$. It holds $F(X,S_X)=\sum_{i\in I} (3b-\frac{3}{2}a_i)+\sum_{i\notin I} (3b-a_i)=3nb-2b-\frac{1}{2}\sum_{i\in I} a_i$ and $F^*(S_X)=\max\{3nb-b,\sum_{i\notin I} 3b+\sum_{i\in I} (3b-a_i)\}=\max\{3nb-b,3nb-\sum_{i\in I} a_i\}$. In consequence
$$Z(X)=\max\{b+\frac{1}{2}\sum_{i\in I} a_i, 2b-\frac{1}{2}\sum_{i\in I} a_i\}$$
 and $Z(X) \leq \frac{3}{2}b$ implies that $\sum_{i\in I} a_i=b$ and, consequently, $I$ forms a partition of $\mathcal{C}$.
\end{proof}

We now provide a simple
 approximation algorithm with a performance ratio of~2.
It  outputs an optimal solution to 
the \textsc{IS} problem with the deterministic vertex weights 
being the midpoints of the
corresponding weights  intervals, i.e.
$\hat{w}_i:=\frac{1}{2}(\underline{w}_i+\overline{w}_i)$ for all $i\in [n]$.
Obviously, its running time is $O(T(n))$, where $T(n)$ is time for solving the \textsc{IS} problem.
\begin{thm} 
\textsc{Max-Min Regret IS} under interval uncertainty representation is approximable within~2. 
\end{thm}
\begin{proof}
The analysis will be  similar  to that in~\cite{KZ06}. The difference is that the underlying deterministic problems discussed in~\cite{KZ06} are minimization ones, whereas the deterministic \textsc{IS} is a maximization problem. So, the result obtained in~\cite{KZ06} cannot be directly applied to \textsc{Min-Max Regret IS}.
Let $\hat{w}_i=\frac{1}{2}(\underline{w}_i+\overline{w}_i)$ for all $i\in [n]$ and 
let $\hat{X}$ be an optimal solution
 for the deterministic weights~$\hat{w}_i$, $i\in [n]$. Let us choose any $X\in \Phi$.  It holds $\sum_{v_i\in \hat{X}}(\underline{w}_i+\overline{w}_i)\geq \sum_{v_i\in X}(\underline{w}_i+\overline{w}_i)$, which implies:
$$
\sum_{v_i\in \hat{X}\setminus X} \overline{w}_i-\sum_{v_i\in X\setminus \hat{X}} \underline{w}_i\geq \sum_{v_i\in X\setminus \hat{X}}\overline{w}_i-\sum_{v_i\in \hat{X}\setminus X}\underline{w}_i.
$$
Therefore, $Z(X)$ fulfills the following inequality:
\begin{equation}
\label{e1}	
	Z(X)\geq F(\hat{X},S_X)-F(X,S_X)= \sum_{v_i\in \hat{X}\setminus X}\overline{w}_i-\sum_{v_i\in X\setminus \hat{X}} \underline{w}_i \geq \sum_{v_i\in X\setminus \hat{X}}\overline{w}_i-\sum_{v_i\in \hat{X}\setminus X}\underline{w}_i.
	\end{equation}
Clearly, $F(\hat{X},S_{\hat{X}})=F(X,S_{\hat{X}})+\sum_{v_i\in \hat{X}\setminus X} \underline{w}_i-\sum_{v_i\in X\setminus \hat{X}} \overline{w}_i$. Hence
$Z(\hat{X})=F^*(S_{\hat{X}})-F(\hat{X},S_{\hat{X}})=F^*(S_{\hat{X}})-F(X,S_{\hat{X}})+\sum_{v_i\in X\setminus \hat{X}} \overline{w}_i-\sum_{v_i\in \hat{X}\setminus X} \underline{w}_i$. Since $Z(X)\geq F^*(S_{\hat{X}})-F(X,S_{\hat{X}})$, the maximal regret of~$\hat{X}$ can be bounded as follows:
\begin{equation}
\label{e2}
Z(\hat{X})\leq Z(X)+\sum_{v_i\in X\setminus \hat{X}} \overline{w}_i-\sum_{v_i\in \hat{X}\setminus X} \underline{w}_i.	
\end{equation}
Inequalities~(\ref{e1}) and~(\ref{e2}) imply $Z(\hat{X})\leq 2Z(X)$ for any $X\in \Phi$, and $Z(\hat{X})\leq 2\cdot opt_2$. 

\begin{figure}[ht]
\centering
\includegraphics{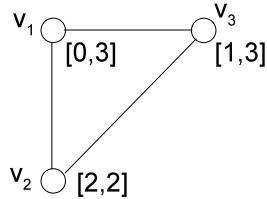}
\caption{A hard example for the 2-approximation algorithm.}\label{fig3}
\end{figure}	

The bound of~2 is tight which is shown in Figure~\ref{fig3}.
The interval graph is a clique composed of 3 vertices. The corresponding interval weights are shown in Figure~\ref{fig3}. The algorithm may return solution  $X=\{v_3\}$.  But $Z(X)=2$ while a trivial verification shows that $opt_2=1$.	
\end{proof}

\section{Conclusions}

In this paper, we have studied the max-min and min-max regret versions of the maximum weighted independent set problem on interval graphs
 with uncertain vertex weights modeled by scenarios.
 We have provided
new complexity and approximation results on the problems, that complete
 the ones previously obtained in the literature.
 For the discrete scenario uncertainty representation,
 we have shown
 that if the number of scenarios $K$ is a part of the input, 
then the max-min version is not at all approximable, the min-max regret version is approximable within $K$ and not approximable within $O(\log^{1-\epsilon}K)$ for any $\epsilon>0$ unless  problems in NP have quasi polynomial algorithms. Furthermore, it has  turned out that both problems admit FTPAS's, when $K$ is constant.
For the interval uncertainty representation, we have proved  that the min-max regret version is NP-hard, 
providing in this way an answer to a question about the complexity of the problem. We have also shown that it is approximable within~2. 
There are still some open questions regarding the min-max regret version of the problem. It would be interesting to provide an approximation algorithm with better than $K$ approximation ratio for the discrete uncertainty representation (when $K$ is part of input) and better than 2 approximation ratio for the interval uncertainty representation. We also do not know whether the latter problem is strongly NP-hard, so it may be solved in pseudopolynomial time and admit and FPTAS.

\subsubsection*{Acknowledgements}
This work was 
partially supported by
 the National Center for Science (Narodowe Centrum Nauki), grant  2013/09/B/ST6/01525.


\end{document}